\documentclass[a4paper,UKenglish,cleveref, autoref, thm-restate]{lipics-v2021}

\usepackage{amsmath}
\usepackage{tikz}
\usetikzlibrary{arrows.meta,decorations.pathreplacing}
\usepackage{xspace}
\newcommand*{\eg}{e.g.\@\xspace}
\newcommand*{\ie}{i.e.\@\xspace}

\usepackage[textwidth=1.9cm,textsize=small]{todonotes}

\usepackage{tikz}
\usetikzlibrary{arrows,shapes,automata, calc, chains, matrix, trees, positioning, scopes}

\title{Stronger bounds on the degree of ambiguity of finite automata} 

\author{Stefan Kiefer}{University of Oxford, United Kingdom}{stefan.kiefer@cs.ox.ac.uk}{}{}
\author{Filip Mazowiecki}{University of Warsaw, Poland}{f.mazowiecki@mimuw.edu.pl}{}{Supported by the Polish National Science Centre SONATA BIS-12 grant number
2022/46/E/ST6/00230}
\author{Andrew Ryzhikov}{University of Warsaw, Poland}{a.ryzhikov@uw.edu.pl}{}{Supported by the Polish National Science Centre SONATA BIS-12 grant number
2022/46/E/ST6/00230}
\author{Antoni Wi\'{s}niewski}{University of Warsaw, Poland}{}{}{Supported by the Polish National Science Centre SONATA BIS-12 grant number
2022/46/E/ST6/00230}

\authorrunning{Stefan Kiefer, Filip Mazowiecki, Andrew Ryzhikov and Antoni Wi\'{s}niewski} 

\Copyright{Jane Open Access and Joan R. Public} 



\category{} 

\relatedversion{} 

\acknowledgements{The authors proved the result of \cref{sec:singleton} without AI assistance. They also found an exponential upper bound for the more general case, where all SCCs have the same size. The general result of \cref{sec:general} was originally proved by ChatGPT 6 Astra. The lower bound in \cref{sec:lower} was suggested earlier by ChatGPT 5.6 Sol, in interaction with the authors. The authors collaborated closely with ChatGPT 6 Astra in writing \cref{sec:lower} and \cref{app:lower}.}

\nolinenumbers 

\EventEditors{John Q. Open and Joan R. Access}
\EventNoEds{2}
\EventLongTitle{42nd Conference on Very Important Topics (CVIT 2016)}
\EventShortTitle{CVIT 2016}
\EventAcronym{CVIT}
\EventYear{2016}
\EventDate{December 24--27, 2016}
\EventLocation{Little Whinging, United Kingdom}
\EventLogo{}
\SeriesVolume{42}
\ArticleNo{23}

\newcommand{\N}{\mathbb{N}}
\newcommand{\Z}{\mathbb{Z}}
\newcommand{\Q}{\mathbb{Q}}

\newcommand{\A}{\mathcal{A}}
\newcommand{\B}{\mathcal{B}}
\newcommand{\Oh}{\mathcal{O}}
\newcommand{\Runs}{\mathsf{Runs}}

\newcommand{\PSPACE}{\textnormal{\textsf{PSPACE}}\xspace}

\AtBeginDocument{\renewcommand{\vec}[1]{\overrightarrow{#1}}}

\newcommand{\set}[1]{\{ #1 \}}
\newcommand{\trans}[1]{\xrightarrow{#1}}

\renewcommand{\vec}[1]{\bm{#1}}

\begin{document}

\maketitle

\begin{abstract}
Ambiguity measures the number of accepting runs in nondeterministic finite automata (NFA). We consider finitely ambiguous NFA, where there exists a constant $N$ such that over every word $w$ there are at most $N$ accepting runs. In such a case we also say that the NFA is $N$-ambiguous. Importantly $N$ depends only on the NFA, it does not depend on the length of the word. Weber and Seidl showed that every NFA is $N$-ambiguous for $N = 2^{\Oh(n \log n)}$, where $n$ is the number of states~\cite{WeberS91}. We improve this to $N = 2^{\Oh(n)}$, which is asymptotically tight.
\end{abstract}

\section{Introduction}
\label{sec:introduction}

Ambiguity is a fundamental syntactic notion in automata theory~\cite{Han2017,ColcombetDagRep2021}. The degree of ambiguity of an NFA $\A$ is a function $f(m) \colon \N \to \N$ defined as follows: $f(m)$ is the maximum, among all words $w$ of length $m$, number of accepting paths labelled by $w$ in $\A$. For example $f(m) = 1$ for deterministic finite automata, and in general $f(m) \le n^{\Oh(m)}$, where $n$ is the size of the NFA. Note that restricting the ambiguity even to $f(m) \le 1$ does not change expressiveness of the class, as every NFA can be determinised.

The growth of ambiguity functions is understood quite well. Namely, given an NFA, its degree of ambiguity is either~$\Oh(m^k)$ for some constant $k$, or is not bounded by any subexponential function~\cite{WeberS91}. This means, for example, that if the degree of ambiguity of an NFA is in $\Oh(\log m)$ or $\Oh(\sqrt{m})$, then it is in~$\Oh(1)$. An NFAs such that its degree of ambiguity $f(m)$ is upper bounded by a constant~$N$ is called $N$-ambiguous; $1$-ambiguous NFAs are usually called unambiguous NFA. An NFA is called finitely ambiguous if it is $N$-ambiguous for some~$N$. Deciding if, for a given NFA, its degree of ambiguity is~$\Oh(m^k)$ for some constant $k$ can be done in polynomial time~\cite{WeberS91}. If the answer is affirmative, computing the smallest such $k$ can also be done in polynomial time~\cite{Allauzen2011}.

The degree of ambiguity can also be defined for more involved automata models, such as weighted automata~\cite{Klimann2004,Filiot2014} and string-to-string transducers~\cite{Weber1993,Filiot2025}, by simply considering the degree of ambiguity of the underlying NFA. Unlike for ordinary Boolean automata, restricting the degree of ambiguity for such models typically restricts their expressivity. As a result, finitely ambiguous automata play a crucial role in algorithmic applications: many problems are decidable for finitely ambiguous automata, but are undecidable for more general classes~\cite{Almagor2022}.
A widely used approach for developing algorithms is then as follows: given an $N$-ambiguous automaton, construct an equivalent automaton that is a union of unambiguous automata. In most constructions, the number of unambiguous automata in such a union is equal to $N$. It is thus natural to ask how large the degree of ambiguity of a finitely ambiguous NFA with $n$ states can be, which is the main question addressed by this paper. 

Weber and Seidl~\cite{WeberS91} showed that it is at most $2^{\Oh(n \log n)}$, and provided a series of $N$-ambiguous $n$-state NFAs with degree of ambiguity $2^{1.0221n}$. A slightly weaker lower bound of $2^{n - 2}$ can be easily obtained by NFAs over a one-letter alphabet. We refer to \cite{WeberS91} for the bounds predating their work. In this paper, we show an upper bound of $2^{\Oh(n)}$, thus obtaining an asymptotically tight bound. Let us briefly note that from the computational complexity viewpoint, the precise ambiguity of a finitely ambiguous NFA is quite hard to analyse. Namely, given a finitely ambiguous NFA and a number $N$, deciding if it is $N$-ambiguous is \PSPACE-complete, even with the promise that $N$ does not exceed the number of states~\cite{ChanI83}.

Ambiguity of NFAs is tightly linked with the growth in (multiplicative) semigroups of nonnegative integer matrices. Consider a finite set of square rational matrices $G = \set{M_1,\ldots, M_k} \subseteq \Q^{n \times n}$, called generators. We define the semigroup generated by $G$ over the rationals: $\mathcal{M} = \set{M \in \Q^{n \times n} \mid M = N_1 \cdot \ldots \cdot  N_m, \text{ where } N_i \in G}$. In general this is a complex object, where natural problems like membership are undecidable~\cite{markov1947certain} (see \eg the introduction of~\cite{PotapovS17}). However, certain approximations of $\mathcal{M}$ are computable, \eg its Zariski closure~\cite{HrushovskiOPW23}. 

A particularly interesting case is when $\mathcal{M}$ is finite, then for example its Zariski closure remains the same set $\mathcal{M}$. Assuming finitness the size of $|\mathcal{M}|$ was analysed already in 70s~\cite{MandelS77} and recently it was shown that $|\mathcal{M}|$ is bounded exponentially~\cite{ManssourGLSW26}. Using our results one can bound the norm of $\mathcal{M}$, \ie bound the entries of matrices in $\mathcal{M}$, provided the generators are 0-1 matrices. Indeed, the assumption about 0-1 matrices allows us to indentify the generators with transition matrices of an NFA. Standard results on weighted automata~\cite[Lemma 5.1, Corollary 5.4]{MandelS77} show that finitness of $\mathcal{M}$ is equivalent to finite ambiguity, and that if the NFA is $N$-ambiguous then $N$ is a bound on the norm of $\mathcal{M}$.


\paragraph*{Technical overview}
A crucial initial step in this work is to divide the states of the NFA into strongly connected components (SCCs) defined by the directed graph of all transitions. Then we deal with two important parameters: $n$ number of states; and $k \le n$ number of components. By definition SCCs form a DAG structure. Weber and Seidl observed that one can assume that the DAG is actually a directed path, where every SCC has exactly one entry state and exactly one exit state. The general case introduces a multiplicative factor of $2^{\Oh(n)}$.

In this simplified form of SCCs forming a path they bound the ambiguity recursively. The idea is to split the path in the middle and recursively use bounds for the left and the right part~\cite[Lemma 2.5]{WeberS91}.
Then a detailed analysis shows that $\A$ is $N$-ambiguous for $N = 2^{\Oh(n \log k)}$. This is asymptotically tight if the number of components is fixed, \eg $k = 2$.

Our idea for a recursive bound is different. Suppose $k > 2$ and fix a word $w$ of length $m$. The key concept is to assign an interval of naturals $[\ell,r]$, for $1 \le \ell \le r \le m$, to every accepting run over $w$. The endpoints of the interval are defined by when the run leaves the first SCC ($\ell$) and when it enters the last SCC ($r$). A crucial observation is that we can effectively bound the number of accepting runs with pairwise disjoint intervals.

Thus we focus on bounding accepting runs such that their intervals have nonempty intersection with the interval of some fixed run. A deeper analysis of this property allows us to split every accepting run into two parts: the first part from the beginning to one of the middle SCCs, the second part leaving from the middle SCC to the last SCC. It is crucial that the middle SCC is not the first or the last ($k$-th) SCC. Then we can recursively bound the number of runs ending in the middle and starting from the middle. Summing over all middle points this recursion behaves similar to Catalan numbers $C_n = \sum_{h=0}^{n-1}C_h\cdot C_{n-h-1}$, which are known to grow exponentially $2^{\Oh(n)}$.

\paragraph*{Organisation}
In \cref{sec:preliminaries} we introduce basic definitions and recall some results of Weber and Seidl~\cite{WeberS91}. In \cref{sec:singleton} we show the solution for the case when every SCC has exactly one state. It is less technical and contains most of the main ideas needed for the general case in \cref{sec:general}. In \cref{sec:lower} we show an example that slightly improves the lower bound for ambiguity to $\Omega((1+\sqrt{2})^n)$, where $\log_2(1+\sqrt{2})=1.271553\ldots$.

\section{Preliminaries}
\label{sec:preliminaries}

Let $s,t \in \Z$. We denote intervals as $[s,t] =\set{i \in \Z \mid s \le i \le t}$. Given a word $w = w_1 \ldots w_m \in \Sigma^*$ and an interval $[i,j] \subseteq [1,m]$ we write $w[i,j]$ for the infix $w_i \ldots w_j$. If $j < 1$ or $i > j$ then $w[i..j] = \epsilon$.

A nondeterministic finite automaton (NFA) is a tuple $\A = (Q, \Sigma, \delta, Q_I, Q_F)$, where: $Q$ is a finite set of states, $\Sigma$ is a finite alphabet, $\delta \subseteq Q \times \Sigma \times Q$ is the transition relation, and $Q_I, Q_F \subseteq Q$ are the initial and final sets of states, respectively. If $Q_I = \set{p}$ or $Q_F = \set{q}$, \ie sets with one state, then we will simply write the state instead, \eg $\A = (Q, \Sigma, \delta, p, q)$.
A \emph{run} over a word $w = w_1 \ldots w_m \in \Sigma^*$ is a sequence of transitions $\rho = (q_1,w_1,q_1'), \ldots (q_m,w_m,q_m')$ such that $q_i' = q_{i+1}$ for all $1 \le i < m$.
We say that a run is an accepting run if $q_1 \in Q_I$ and $q_m' \in Q_F$.
If a run exists we write $q_1\trans{w}q_m'$, in this case we also say that $\rho$ witnesses $q_1\trans{w}q_m'$.
We also write $q_1\trans{*}q_m'$ if we mean that there exists $w$ such that $q_1\trans{w}q_m'$. We will abuse this notation writing \eg $q_1\trans{w_1w_2w_3}q_3' \trans{w_4} q_4' \trans{w_5w_6}q_6'$ meaning that there are three runs. Runs over the empty word $\epsilon$ are allowed and $q_1\trans{\epsilon}q_m'$ if and only if $q_1 = q_m'$.
%

Given a word $w$ we write $\Runs_{\A}(w)$ (or $\Runs(w)$ if $\A$ is clear from the context) for the set of accepting runs in $\A$ over $w$. We say that $\A$ is \emph{finitely ambiguous} if there exists $N$ such that for all words $w\in \Sigma^*$ we have $|\Runs(w)| \le N$. Importantly, $N$ is uniform for all words $w$. We will also write that $\A$ is $N$-ambiguous in this case.
In~\cite{WeberS91} it was shown that $\A$ is $N$-ambiguous for $N \le n^n = \Oh(2^{n\log n})$. In \cref{sec:lower} we discuss known exponential lower bounds for $N$.
The goal of this paper is to improve the upper bound to $N = 2^{\Oh(n)}$.
We start by introducing notation and results from \cite{WeberS91}, that simplify the input.

A related concept to finite ambiguity is the IDA property, introduced in~\cite{WeberS91}. We say that $\A$ \emph{complies with IDA} if there exist states $p,q \in Q$, where $p \neq q$, and a word $w \in \Sigma^*$ such that: $p \trans{w} p \trans{w} q \trans{w} q$.

To use the IDA property we need one more definition.
We say that the NFA $\A$ is \emph{trim} if all states $q \in Q$ are reachable and co-reachable, \ie there are $q_I \in Q_I$ and
$q_F \in Q_F$ such that $q_I \trans{*} q$ and $q \trans{*} q_F$. We will usually assume that NFAs are trim, as removing states that are not reachable or co-reachable is simple and does not change the ambiguity of the automaton.

\begin{theorem}[{\cite[Theorem 3.1]{WeberS91}}]\label{theorem:ida}
A trim NFA is finitely ambiguous if and only if it does not comply with IDA.
\end{theorem}

One corollary from \cref{theorem:ida} is that finite ambiguity does not depend on initial and final states for trim NFA. We will implicitly use this property when changing the initial and final states in finitely ambiguous NFAs.

Let $Q = Q_1 \uplus \ldots \uplus Q_k$ be a decomposition of the states into strongly connected components (SCCs). We will always assume that $Q_1$, \ldots $Q_k$ are ordered topologically. That is, if there is a transition $(p,a,q) \in \delta$ such that $p \in Q_i$, $q \in Q_j$ for $i \neq j$ then $i < j$. A transition between different connected components is called a \emph{bridge}.

We will use the following simple claim about finite ambiguity in SCCs.

\begin{claim}[{\cite[Lemma 2.4]{WeberS91}}]\label{claim:unambiguous}
Let $\A$ be a trim finitely ambiguous NFA over the alphabet $\Sigma$. Fix $w \in \Sigma^*$ and two states $p$, $q$ from the same SCC. There is at most one run witnessing $p \trans{w} q$.
\end{claim}

We observe that one implication from \cref{theorem:ida} can be slightly generalised, which will be convenient in \cref{sec:general}.
Figure~\ref{fig:idasubset} illustrates the three conditions in the following proposition.

\begin{proposition}\label{proposition:idasubset}
Fix a trim NFA $\A$ and let $Q_i$ and $Q_j$ be its SCCs such that $i < j$. Fix two nonempty subsets $U \subseteq Q_i$, $V \subseteq Q_j$. Suppose there is a word $w$ such that:
\begin{enumerate}
 \item\label{p:1} for all $u \in U$ there is $u'\in U$ such that $u' \trans{w} u$;
 \item\label{p:2} for all $v \in V$ there is $v' \in V$ such that $v \trans{w} v'$;
 \item\label{p:3} there exist $u \in U$ and $v \in V$ such that $u \trans{w} v$.
\end{enumerate}
Then $\A$ is not finitely ambiguous.
\end{proposition}

\begin{figure}[htbp]
\begin{nolinenumbers}
\centering
\begingroup
\definecolor{subsetink}{HTML}{263747}
\begin{tikzpicture}[
  x=1cm,y=1cm,font=\small,text=subsetink,
  patharrow/.style={-{Stealth[length=2.5mm,width=1.7mm]},line width=1.25pt,shorten <=2.8pt,shorten >=3.2pt},
  wordlabel/.style={font=\small,fill=white,inner sep=1.5pt},
  dot/.style={circle,inner sep=0pt,minimum size=4.6pt},
  marked dot/.style={dot,draw=white,line width=.65pt},
  subset band/.style={fill=yellow!12,rounded corners=5pt},
  subset brace/.style={decorate,decoration={brace,amplitude=5pt},gray!65}]
\fill[subset band] (-.22,2.67) rectangle (5.42,5.03);
\fill[subset band] (-.22,-.73) rectangle (5.42,1.63);
\foreach \yy in {4.7,3.85,3,1.3,.45,-.4}
  \draw[gray!22] (0,\yy)--(5.2,\yy);
\foreach \xx in {0,5.2}
  \draw[gray!18,densely dotted] (\xx,-.73)--(\xx,5.03);
\foreach \side/\xx in {L/0,R/5.2}{
  \foreach \state/\yy in {uone/4.7,u/3.85,utwo/3,v/1.3,vone/.45,vtwo/-.4}{
    \coordinate (\state\side) at (\xx,\yy);
    \node[dot,fill=gray!40] at (\state\side) {};
  }
}
\foreach \statelabel/\yy in {u_1/4.7,u/3.85,u_2/3,v/1.3,v_1/.45,v_2/-.4}{
  \node[anchor=east] at (-.32,\yy) {$\statelabel$};
  \node[anchor=west] at (5.52,\yy) {$\statelabel$};
}
\draw[subset brace] (-.83,2.67)--(-.83,5.03)
  node[midway,left=8pt,align=center,text=subsetink] {$U$\\[-1pt]\scriptsize inside $Q_i$};
\draw[subset brace] (-.83,-.73)--(-.83,1.63)
  node[midway,left=8pt,align=center,text=subsetink] {$V$\\[-1pt]\scriptsize inside $Q_j$};
\foreach \from/\to in {uoneL/uoneR,uoneL/uR,uL/utwoR}
  \draw[patharrow,black!60] (\from)--(\to) node[pos=.55,above,wordlabel] {$w$};
\foreach \from/\to in {vL/voneR,voneL/vtwoR,vtwoL/vtwoR}
  \draw[patharrow,black!60] (\from)--(\to) node[pos=.43,above,wordlabel] {$w$};
\draw[patharrow,subsetink] (uL)--(vR)
  node[pos=.53,above,wordlabel] {$(\ref{p:3})\quad w$};
\foreach \state in {uoneL,uoneR,uR,utwoR,vL,voneL,vtwoL,voneR,vtwoR}
  \node[marked dot,fill=black!60] at (\state) {};
\foreach \state in {uL,vR}
  \node[marked dot,fill=subsetink] at (\state) {};
\node[anchor=west,align=left] at (6.3,3.85)
  {$(\ref{p:1})$ Every state in $U$ can be\\
   reached from a state in $U$\\
   by a run labelled $w$.};
\node[anchor=west,align=left] at (6.3,.45)
  {$(\ref{p:2})$ Every state in $V$ has a\\
   run labelled $w$ to a state\\
   in $V$.};
\end{tikzpicture}
\endgroup
\caption{Illustration of the hypotheses of Proposition~\ref{proposition:idasubset}.
Each arrow represents a run labelled $w$.}
\label{fig:idasubset}
\end{nolinenumbers}
\end{figure}
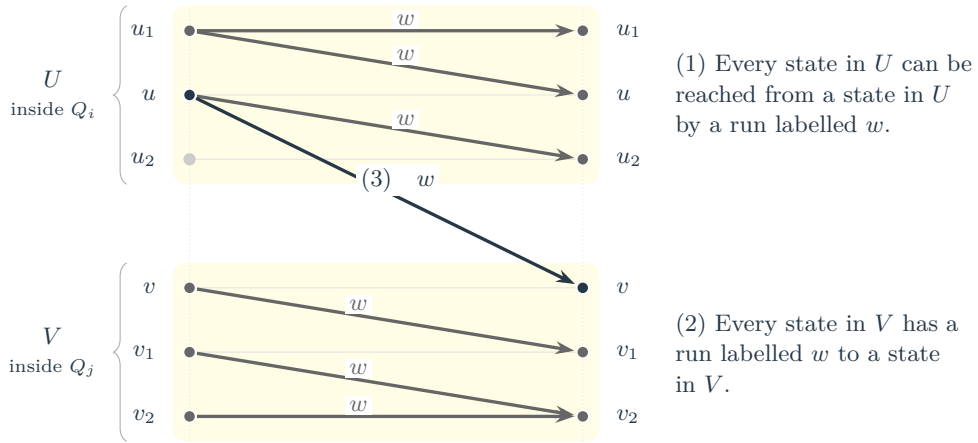

\begin{proof}
Let $u$ and $v$ be as in \cref{p:3}. We observe that for all $m$ there are $u_{m+1},\ldots,u_{1} \in U$; and $v_1,\ldots,v_{m+1}\in V$ such that $u_{1} \trans{w} u_2 \trans{w} \ldots u_{m+1} \trans{w} v_1 \trans{w} v_2 \trans{w} \ldots v_{m+1}$. Indeed, defining $u_{m+1} =u$ and $v_1 = v$ we guarantee the middle transitions by \cref{p:3}. Then we can extend the path both ways: applying \cref{p:1} for states in $U$ and \cref{p:2} for states in $V$.

Fix $m > |Q_i|\cdot |Q_j|$. Among the $m+1$ pairs $(u_t,v_t)$, two must coincide, so there exist $1 \le a < b \le m+1$ such that $u_a = u_b$ and $v_a = v_b$.
Choose $a\le c<b$ and an integer $N\ge1$ such that $m+1+c-a=N(b-a)$; this is possible because $m+1+c-a$ ranges over $b-a$ consecutive integers.
The constructed run gives $u_a\trans{w^{N(b-a)}}v_c$.
There are also runs from $u_a$ and $v_c$ back to themselves labelled $w^{b-a}$: follow $u_a,\ldots,u_b=u_a$ and $v_c,\ldots,v_b=v_a,\ldots,v_c$, respectively.
Repeating these runs $N$ times gives $u_a \trans{w^{N(b-a)}} u_a \trans{w^{N(b-a)}} v_c \trans{w^{N(b-a)}} v_c$, so $\A$ complies with IDA.
\end{proof}

We now define restricted NFAs with which it will be more convenient to work.
We say that $\A$ is a \emph{chain NFA} if there are states $p_1, q_1 \in Q_1$; $p_2, q_2 \in Q_2$; \ldots $p_k, q_k \in Q_k$ such that:
\begin{itemize}
 \item the initial and final states in $\A$ are: $Q_I = \set{p_1}$ and $Q_F = \set{q_k}$;
 \item all bridges are of the form $(q_i,a,p_{i+1})$ for some $a \in \Sigma$ and $1 \le i < k$.
\end{itemize}
We say that $p_i$ and $q_i$ are the states characterising the chain. Notice that accepting runs in a chain NFA have exactly $k-1$ bridges.

Most of the time we will assume that the input is a chain NFA. This is due to the following reduction, which introduces a multiplicative factor of $2^{\Oh(n)}$ in ambiguity compared to the general case.

\begin{lemma}[Corollary of {\cite[Lemma 2.3]{WeberS91}}]\label{lemma:chain}
Let $\A$ be a finitely ambiguous NFA with $n$ states. There are finitely ambiguous chain NFAs $\A_1, \ldots, \A_N$ such that: $|\Runs_{\A}(w)| = \sum_{i=1}^N |\Runs_{\A_i}(w)|$; every $\A_i$ has at most $n$ states; and $N \le 5^{n/2}$.
\end{lemma}

Let $\A$ be a chain NFA with $k\ge2$ and the SCC decomposition: $Q = Q_1 \uplus \ldots \uplus Q_k$. Fix $w = a_1\ldots a_m \in \Sigma^*$ and let $\rho = t_1 \ldots t_{m} \in \Runs_{\A}(w)$.
We will assign an interval to each run. First, recall that $p_i$, $q_i$ are states characterising the chain. To define the left endpoint let $\ell_{\A} \colon \Runs_{\A}(w) \to \set{1,\ldots,m}$ be the last time $\rho$ was in $Q_1$. Formally, $\ell_{\A}(\rho) = s$ if $t_{s} = (q_1,a_{s},p_2)$. Similarly, $r_{\A} \colon \Runs_{\A}(w) \to \set{1,\ldots,m}$ is the first time $\rho$ entered $p_k$. Formally, $r_{\A}(\rho) = s$ if $t_{s} = (q_{k-1},a_s,p_k)$. We will drop the automaton subscript if it is clear from context. Thus to every $\rho$ we assign the interval $I_\rho = [\ell(\rho),r(\rho)]$.

\section{Singleton SCCs}
\label{sec:singleton}

In this section we show the upper bound for a simple variant of chain NFAs. It illustrates the main ideas in the general construction, while it is conceptually simpler. We will exploit this to provide intuition in \cref{sec:general}.

A chain NFA $\A = (Q, \Sigma, \delta, p_1, p_n)$ is called a \emph{singleton-SCC NFA} if: $Q = \set{p_1,\ldots,p_n}$; $(p_i,a,p_j) \in \delta$ then $j \in \set{i,i+1}$. So every transition is either a self loop or moves one step to the right, see \cref{fig:singletons} for an example. The following is the main result of this section.

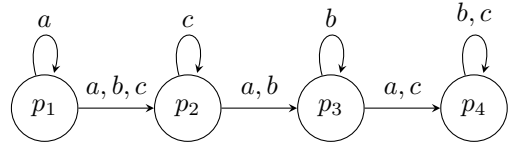
\begin{figure}
\centering
\begin{tikzpicture}
\node[state] (p1) {$p_1$};
\node[state,right = 1cm of p1] (p2) {$p_2$};
\node[state,right = 1cm of p2] (p3) {$p_3$};
\node[state,right = 1cm of p3] (p4) {$p_4$};

\path
(p1) edge[->,>=stealth] node[above] {$a,b,c$} (p2)
(p2) edge[->,>=stealth] node[above] {$a,b$} (p3)
(p3) edge[->,>=stealth] node[above] {$a,c$} (p4)
(p1) edge[->,>=stealth,loop above] node[above] {$a$} (p1)
(p2) edge[->,>=stealth,loop above] node[above] {$c$} (p2)
(p3) edge[->,>=stealth,loop above] node[above] {$b$} (p3)
(p4) edge[->,>=stealth,loop above] node[above] {$b,c$} (p4)
;

\end{tikzpicture}
\caption{A $4$ state chain NFA with singleton SCCs over the alphabet $\Sigma = \set{a,b,c}$.}\label{fig:singletons}
\end{figure}

\begin{theorem}\label{theorem:singletons}
A finitely ambiguous singleton-SCC NFA with $n$ states is $N$-ambiguous for $N = 2^{\Oh(n)}$.
\end{theorem}

\begin{example}
Consider the automaton $\A$ in \cref{fig:singletons}. Let $w = aaacccbbccc$. There are three accepting runs in $\Runs_{\A}(w) = \set{\rho_1, \rho_2, \rho_3}$: $\rho_1$ moves always to the right reading $a$; $\rho_2$ stays in $p_1$ reading the first two $a$'s then moves to $p_2$ after reading the last $a$; finally $\rho_3$ stays in $p_1$ reading all three $a$'s. The intervals assigned to them are: $I_{\rho_1} = [1,3]$, $I_{\rho_2} = [3,9]$, $I_{\rho_3} = [4,9]$.
\end{example}

\begin{lemma}\label{lemma:smallm}
Let $\A$ be a trim finitely ambiguous singleton SCC NFA with at least two states and fix a word $w \in \Sigma^*$. Let $X \subseteq \Runs_{\A}(w)$ such that the intervals assigned to runs in $X$ are pairwise disjoint, \ie for all $\rho_1, \rho_2 \in X$ if $\rho_1 \neq \rho_2$ then $I_{\rho_1} \cap I_{\rho_2} = \emptyset$. Then $|X| \le 2$.
\end{lemma}

\begin{proof}
Towards a contradiction suppose there are $\rho_1, \rho_2, \rho_3 \in \Runs_{\A}(w)$ such that $I_{\rho_1} = [l_1,r_1]$, $I_{\rho_2} = [l_2, r_2]$, and $I_{\rho_3} = [l_3,r_3]$ are pairwise disjoint. Without loss of generality assume that $l_1 \le r_1 < l_2 \le r_2 < l_3 \le r_3$. Let $p_1$ and $p_n$ be the initial and final states in $\A$, respectively. Let $v = w[l_2,r_2]$. We observe that: $p_1 \trans{v} p_1$ (part of $\rho_3$); $p_1 \trans{v} p_n$ (part of $\rho_2$); and $p_n \trans{v} p_n$ (part of $\rho_1$). This shows that $\A$ complies with IDA, which is a contradiction with finite ambiguity.
\end{proof}

For every $n$ we define $A_n \in \N$ as the minimal $N$ such that all trim finitely ambiguous singleton SCC NFA with $n$ states are $N$-ambiguous. The following lemma allows to bound $A_n$ recursively.

\begin{lemma}\label{lemma:Anrecursion}
We have $A_1 \le 1$, $A_2 \le 2$, and $A_n \le 3\sum_{i = 2}^{n-1} A_i \cdot A_{n-i+1}$ for $n > 2$.
\end{lemma}

\begin{proof}
For $n = 1$ this follows from \cref{claim:unambiguous}. For $n=2$ it follows from \cref{lemma:smallm}, one needs to observe that all accepting runs have disjoint intervals. This is because transitions within both SCCs are deterministic, so the only nondeterminism comes from changing the SCC. In fact $A_2=2$: consider the automaton with initial state $p_1$, final state $p_2$, an $a$-loop on $p_1$, a $b$-loop on $p_2$, and transitions labelled $a$ and $b$ from $p_1$ to $p_2$. It has at most two accepting runs on any word and exactly two on $ab$.

Suppose $n >2$ and let $\A = (Q, \Sigma, \delta, p_1, p_n)$ be a trim finitely ambiguous singleton SCC NFA with $n$ states $Q = \set{p_1, \ldots, p_n}$ and $w$ a word such that $|\Runs_{\A}(w)| = A_n$. Choose a run $\rho_1 \in \Runs_{\A}(w)$ such that the right end of its interval is minimal, \ie if $I_{\rho_1} = [l_1,r_1]$ then for all other $I_{\rho} = [l,r] \in \Runs_{\A}(w)$ we have $r_1 \le r$. Let $\Runs_{\A}^{1}(w) \subseteq \Runs_{\A}(w)$ be the set of runs $\rho$ such that $I_{\rho} \cap I_{\rho_1} \neq \emptyset$.


We define $\Runs_{\A}^{2}(w) = \Runs_{\A}(w) \setminus \Runs_{\A}^1(w)$. If this set is empty, set $\Runs_{\A}^{2,i}(w)=\emptyset$ for $2\le i\le n$ and omit the subsequent steps involving $\rho_2$ and $r_2$. Otherwise, choose a run $\rho_2 \in \Runs_{\A}^2(w)$ such that the right end of its interval is minimal, \ie if $I_{\rho_2} = [l_2,r_2]$ then for all other $I_{\rho} = [l,r] \in \Runs^2_{\A}(w)$ we have $r_2 \le r$. Observe that $l_1 \le r_1 < l_2 \le r_2$. By \cref{lemma:smallm} for all $\rho \in \Runs^2_{\A}(w)$ we have $I_{\rho} \cap I_{\rho_2} \neq \emptyset$.

Let $\rho = t_1 \ldots t_{|w|} \in \Runs_{\A}^1(w)$. We decompose the set $\Runs_{\A}^1(w) = \biguplus_{i = 2}^n \Runs_{\A}^{1,i}(w)$ depending on the state reached in position $r_1$. Formally, $\rho \in \Runs_{\A}^{1,i}(w)$ if and only if $t_{r_1} = (p,a,q)$ for $q = p_i$. Note that $i \ge 2$. This is because $I_{\rho} \cap I_{\rho_1} \neq \emptyset$, which by minimality of $r_1$ implies that if $I_{\rho} = [l,r]$ then $l \le r_1$. Similarly, we decompose $\Runs_{\A}^2(w) = \biguplus_{i = 2}^n \Runs_{\A}^{2,i}(w)$ depending on the state reached in position $r_2$.


%

Let $1 \le s \le t \le n$. We write $\A_{s \to t}$ for the automaton $(\set{p_s,\ldots,p_t}, \Sigma, \delta', p_s, p_t)$ with $t-s+1$ states, where $\delta'$ is
inherited from $\delta$. One should think it is like trimming $\A$ to an infix of states bounded by $s$ and $t$.

Let $2 \le i <n$. Observe that every run in $\Runs_{\A}^{1,i}(w)$ can be decomposed into two runs: $\Runs_{\A_{1 \to i}}(w[1..r_1])$ and $\Runs_{\A_{i\to n}}(w[r_1+1..|w|])$. Similarly runs in $\Runs_{\A}^{2,i}(w)$ can be decomposed into: $\Runs_{\A_{1 \to i}}(w[1..r_2])$ and $\Runs_{\A_{i\to n}}(w[r_2+1..|w|])$. We conclude that:
\begin{align*}
|\Runs_{\A}^{1,i}(w)|, |\Runs_{\A}^{2,i}(w)| \le A_i \cdot A_{n-i+1}.
\end{align*}

It remains to bound $|\Runs_{\A}^{1,n}(w)|$ and $|\Runs_{\A}^{2,n}(w)|$. For this we use the minimality of $r_1$ and $r_2$. Note that for every $\rho = t_1\ldots t_{|w|} \in \Runs_{\A}^{1,n}(w)$ we have $t_{r_1} = (p_{n-1},w_{r_1},p_n)$. This also implies that $t_{s} = (p_n,w_s,p_n)$ for all $s > r_1$. We conclude that
$$
|\Runs_{\A}^{1,n}(w)| = |\Runs_{\A_{1\to n-1}}(w[1..r_1-1])| \le A_{n-1}.
$$
Similarly, observe that
$$
|\Runs_{\A}^{2,n}(w)| \le |\Runs_{\A_{1\to n-1}}(w[1..r_2-1])| \le A_{n-1}.
$$
Combining everything we get 
$$
A_n = |\Runs_{\A}(w)| \le 2\sum_{i = 2}^{n-1}A_i \cdot A_{n-i+1} + 2A_{n-1} \le 3\sum_{i = 2}^{n-1}A_i \cdot A_{n-i+1},
$$
as required. For the last inequality observe that $A_{n-1} \cdot A_2=2A_{n-1}$ appears in the sum.
\end{proof}

\begin{proof}[Proof of \cref{theorem:singletons}]
It suffices to estimate $A_n$ for $n > 2$. By \cref{lemma:Anrecursion} we have: $A_n \le 3\sum_{i = 2}^{n-1} A_i \cdot A_{n-i+1}$. 
We show by induction that for $n \ge 2$ we have $A_{n} \le 3^{2n-3}\cdot 2^{n-1} \cdot C_{n-2}$, where $C_n$ is the $n$-th Catalan number. Indeed, for $n = 2$ by \cref{lemma:Anrecursion} we have $A_n \le 2 \le 3^12^1C_0 = 6$. For the inductive step let $n > 2$, then
\begin{align*}
A_n \le 3\sum_{i = 2}^{n-1} A_i \cdot A_{n-i+1} \le 3\sum_{i = 2}^{n-1} 3^{2i - 3} 2^{i-1}C_{i-2} \cdot 3^{2(n-i+1) - 3} 2^{n-i} C_{n-i-1}\\
= 3^{2n-3}2^{n-1} \sum_{i = 0}^{n-3}C_i \cdot C_{n-i-3} = 3^{2n-3}2^{n-1}C_{n-2}.
\end{align*}
This concludes the proof as it is well-known that for Catalan numbers $C_n = 2^{\Oh(n)}$.
\end{proof}

\section{General case}
\label{sec:general}

Fix a trim finitely ambiguous chain NFA $\A = (Q, \Sigma, \delta, p_1, q_k)$ with strongly connected components $Q = Q_1 \uplus Q_2 \uplus \ldots \uplus Q_k$. Let $p_1, q_1 \in Q_1$, \ldots, $p_k,q_k \in Q_k$ be the states characterising the chain. We will denote $|Q| = n$, $|Q_i| = n_i$, in particular $n = \sum_{i = 1}^k n_i$.
Our goal is to prove the general upper bound on the ambiguity of $\A$.

\begin{theorem}\label{theorem:main}
A finitely ambiguous chain NFA with $n$ states is $N$-ambiguous for $N = 2^{\Oh(n)}$.
\end{theorem}

In the following definitions and lemmas, assume $k\ge2$; the one-component case of the theorem follows from \cref{claim:unambiguous}.
For every word $w \in \Sigma^*$ we will decompose the set of accepting runs, fixing some reachable and co-reachable states in the first and last component
$$
\Runs_{\A}(w) = \bigcup_{\emptyset \neq U\subseteq Q_1,\; \emptyset \neq V \subseteq Q_k} \Runs_{\A}^{U,V}(w).
$$
Formally, $\rho \in \Runs_{\A}^{U,V}(w)$ if and only if given the assigned interval $I_{\rho} = [l,r]$ we have
$$
U = \set{q \in Q_1 \mid p_1 \trans{w[1..l-1]} q}, \text{ and } V = \set{p \in Q_k \mid p \trans{w[r+1..|w|]}q_k}.
$$
Note that $U$ and $V$ are nonempty as we have $q_1 \in U$ and $p_k \in V$.
For such sets of runs we generalise \cref{lemma:smallm}.

\begin{lemma}\label{lemma:fulldecomposition}
Let $\A$ be a trim finitely ambiguous chain NFA and let $Q_1$ and $Q_k$ be the first and last SCC, respectively, where $|Q_1| = n_1$ and $|Q_k| = n_k$. Fix $U \subseteq Q_1$ and $V \subseteq Q_k$ and let $X \subseteq \Runs_{\A}^{U,V}(w)$ such that intervals assigned to runs in $X$ are pairwise disjoint. Then $|X| \le 2^{\min (n_1, n_k)}$.
\end{lemma}

\begin{proof}
Let $I_{\rho_1} = [l_1,r_1]$, $I_{\rho_2} = [l_2, r_2]$, \ldots $I_{\rho_s} = [l_s,r_s]$, where $\rho_1,\ldots, \rho_s \in X$ and $l_1 \le r_1 < l_2 \le r_2 < \ldots < l_s \le r_s$. It suffices to show that $s \le 2^{n_k}$ and that $s \le 2^{n_1}$.

To prove the first inequality for every $1 \le i \le s$ we define $V_{i} = \set{p \in Q_k \mid p \trans{w[l_i..|w|]}q_k}$. Towards a contradiction assume that $s > 2^{n_k}$. Then there are $1 \le i < j \le s$ such that $V_i = V_j$. Let $x = w[l_i \ldots l_j -1]$. The idea is to invoke \cref{proposition:idasubset} for $U \subseteq Q_1$, $V_i \subseteq Q_k$ and the word $x$; item 3 below also shows that $V_i$ is nonempty. Indeed,
\begin{enumerate}
 \item Let $u \in U$. By definition $U = \set{q \in Q_1 \mid p_1 \trans{w[1..l_j-1]} q}$, so there is a run $\gamma$ witnessing $p_1 \trans{w[1..l_j-1]}u$. But since $U = \set{q \in Q_1 \mid p_1 \trans{w[1..l_i-1]} q}$ it must be that $\gamma = \gamma_1 \gamma_2$, where $\gamma_1$ is a run $p_1 \trans{w[1..l_i-1]} u'$ with $u' \in U$. Then $\gamma_2$ witnesses $u' \trans{x} u$ as required.
 \item Let $v \in V_i = \set{p \in Q_k \mid p \trans{w[l_i..|w|]}q_k}$ so there is a run $\gamma$ witnessing $v \trans{w[l_i \ldots |w|]}q_k$. Since $V_i = V_j$ we have $\gamma = \gamma_1\gamma_2$, where $\gamma_2$ witnesses $v' \trans{w[l_j\ldots |w|]} q_k$ for $v' \in V_i$. Thus $\gamma_1$ witnesses $v \trans{x} v'$.
 \item Consider the run $\rho_i$ and recall that its infix witnesses $q_1 \trans{w[l_i..r_i]}p_k$. Since $r_i < l_j$, the same run witnesses $q_1 \trans{x}p$ for some $p \in Q_k$. Its remaining suffix shows that $p \in V_j = V_i$, and $q_1 \in U$ by definition.
\end{enumerate}
By \cref{proposition:idasubset} $\A$ is not finitely ambiguous, which is a contradiction.

The proof that $s \le 2^{n_1}$ is analogous. One must define reachable sets $U_i \subseteq Q_1$ for words $w[1..r_i]$ and conclude that these sets are pairwise distinct.
%
\end{proof}

We now define an analog of the sequence $A_n$ in \cref{sec:singleton}. It requires to be more detailed as we need to take into account the sizes of SCCs.
For this we will need to generalise the definition of automata trimmed to some of the connected components from \cref{sec:singleton}.
Consider a trim finitely ambiguous chain NFA $\A = (Q, \Sigma, \delta, p_1, q_k)$ with strongly connected components $Q = Q_1 \uplus Q_2 \uplus \ldots \uplus Q_k$.
Given $1 \le i \le j \le k$, $p \in Q_i$, and $q \in Q_j$ we define $\A_{p \to q} = (\bigcup_{h = i}^j Q_h,\Sigma,\delta',p,q)$, where $\delta'$ is inherited from $\delta$. Essentially we choose new initial and final states and trim the previous and latter components.

\begin{definition}\label{definition:Nij}
Fix $k\ge2$ and $n_1, \ldots, n_k > 0$ and denote the vector $\vec{n} = (n_1,\ldots,n_k)$.
Consider all trim finitely-ambiguous chain NFA $\B$ such that their SCCs are ordered $Q_1$, \ldots, $Q_k$ and $|Q_i| = n_i$ for all $1 \le i \le k$. We define $N_{\vec{n}}[1,k]$ to be the maximal value of $|\Runs^{U,V}_{\B}(w)|$ among all such automata $\B$, words $w$ over the alphabet of $\B$, and nonempty subsets $U \subseteq Q_1$, $V \subseteq Q_k$.

Furthermore, for every $1 \le i < j \le k$ we define $N_{\vec{n}}[i,j]$ as the maximal value of $|\Runs^{U,V}_{\B_{p \to q}}(w)|$ among all such automata $\B$, words $w$ over the alphabet of $\B$, nonempty subsets $U \subseteq Q_i$, $V \subseteq Q_j$, and $p \in Q_i$, $q \in Q_{j}$.
When clear we drop the subscript $\vec{n}$ and simply write $N[i,j]$.
\end{definition}

The notation $N_{\vec{n}}[i,j]$ is for convenience. By definition it is equivalent to $N_{\vec{m}}[1,j-i+1]$, where $\vec{m} = (n_i,\ldots,n_j)$.
Like in \cref{lemma:Anrecursion} we show a recursive bound on $N[1,k]$.

\begin{lemma}\label{lemma:recursion}
Fix $k\ge2$ and $n_1, \ldots, n_k > 0$ and let $\vec{n} = (n_1,\ldots,n_k)$.
If $k = 2$ then $N_{\vec{n}}[1,2] \le 2^{\min(n_1, n_2)}$. If $k>2$ then $N_{\vec{n}}[1,k] \le 2 \cdot 2^{\min(n_1,n_k)} \cdot \sum_{h=2}^{k-1} n_{h} \cdot 2^{2 n_{h}} \cdot N_{\vec{n}}[1,h]\cdot N_{\vec{n}}[h,k]$.
\end{lemma}

\begin{proof}
For the case $k=2$ consider the automaton $\A = (Q_1 \cup Q_2, \Sigma, \delta, p_1, q_2)$ with two SCCs $Q_1$ and $Q_2$ of sizes $n_1$ and $n_2$, respectively. Let $p_1,q_1 \in Q_1$ and $p_2,q_2 \in Q_2$ be the states characterising the chain. Observe that for every $\rho \in \Runs^{U,V}_{\A}(w)$ the assigned interval is of the form $I_{\rho} = [a,a]$. We prove that if $\rho, \rho' \in \Runs^{U,V}_{\A}(w)$ and $I_{\rho} = I_{\rho'} = [a,a]$ then $\rho = \rho'$. Indeed, otherwise there are either two different runs witnessing $p_1 \trans{w[1..a-1]} q_1$ or two different runs witnessing $p_2 \trans{w[a+1..|w|]}q_2$. This is a contradiction by \cref{claim:unambiguous}. Thus the bound $2^{\min(n_1, n_2)}$ follows directly from \cref{lemma:fulldecomposition}.

We move to the most interesting case $k > 2$.
Let $\B$, $w$, $U \subseteq Q_1$, $V \subseteq Q_k$ be such that $|\Runs^{U,V}_{\B}(w)| = N[1,k]$.
Denote $Q_1$, $Q_{2}$, \ldots $Q_k$ the consecutive SCCs in $\B$ of sizes $n_1$, \ldots, $n_k$, respectively. And let $p_1 \in Q_1$ and $q_k \in Q_k$ be the initial and final states, respectively.

We decompose $\Runs^{U,V}_{\B}(w) = \biguplus_{s = 1}^m\Runs^{U,V,s}_{\B}(w)$ as follows. Let $\rho_1 \in \Runs^{U,V}_{\B}(w)$ such that it has the smallest right endpoint in the assigned interval $I_{\rho_1} = [l_1,r_1]$. Formally, for all $\rho \in \Runs^{U,V}_{\B}(w)$ if $I_{\rho} = [l,r]$ then $r_1 \le r$. Then $\rho \in \Runs^{U,V,1}_{\B}(w)$ if $I_{\rho} \cap I_{\rho_1} \neq \emptyset$. Repeat this procedure for consecutive $s > 1$: among the remaining runs choose $\rho_s \in \Runs^{U,V}_{\B}(w) \setminus \bigcup_{s' < s}\Runs^{U,V,s'}_{\B}(w)$ with the smallest right endpoint in the interval $I_{\rho_s} = [l_s,r_s]$; and define $\Runs^{U,V,s}_{\B}(w)$ as all remaining runs $\rho$ such that $I_{\rho} \cap I_{\rho_s} \neq \emptyset$. Observe that by \cref{lemma:fulldecomposition} we have $m \le 2^{\min(n_1,n_k)}$.

We now decompose further each $\Runs^{U,V,s}_{\B}(w) = \biguplus_{h = 2}^{k} \Runs^{U,V,s,h}_{\B}(w)$ depending on which SCC the run visits after reading $r_s$ letters. Let $\rho = t_1 \ldots t_{|w|} \in \Runs^{U,V,s}_{\B}(w)$ and let $t_{r_s} = (p,w_{r_s},q)$. Then $\rho \in \Runs^{U,V,s,h}$ if and only if $q \in Q_{h}$. Note that $h \ge 2$. This is because $I_\rho \cap I_{\rho_s} \neq \emptyset$, which by minimality of $r_s$ implies that if $I_{\rho} = [l,r]$ then $l \le r_s$.

We now bound $\Runs^{U,V,s,h}_{\B}(w)$ for the case $2 \le h < k$. The idea is to split the run into a prefix ending in $Q_h$ and a suffix starting there. 
For this we need to compute new $U'$ and $V'$.
Given $\rho = t_1 \ldots t_{|w|} \in \Runs^{U,V,s,h}_{\B}(w)$ let $a \le r_s < b$ be the indices of the bridges entering and exiting $Q_h$, \ie $t_{a} = (q_{h-1},w_{a},p_h)$ and $t_{b} = (q_h, w_{b}, p_{h+1})$.
Let $q \in Q_h$ be such that $t_{r_s} = (p,w_{r_s},q)$, and define
\[
V' = \set{p \in Q_h \mid p \trans{w[a+1..r_s]}q},
\qquad
U' = \set{p \in Q_h \mid q \trans{w[r_s+1..b-1]}p}.
\]
The run $\rho$ shows that $p_h \in V'$ and $q_h \in U'$, so both sets are nonempty.

Then we decompose $\rho = \rho_1\rho_2$, where: $\rho_1 \in \Runs^{U,V'}_{\B_{p_1 \to q}}(w[1..r_s])$ and $\rho_2 \in \Runs^{U',V}_{\B_{q \to q_k}}(w[r_s+1..|w|])$. The number of choices for $V'$ and $U'$ are both bounded by $2^{n_h}$, and the number of choices for $q$ is $n_h$. Altogether the bound on the number of $\rho_1$ and $\rho_2$ gives us
$$
|\Runs^{U,V,s,h}_{\B}(w)| \le n_h2^{2n_h} \cdot N[1,h]\cdot N[h,k].
$$

It remains to consider $h = k$. Let $\rho = t_1 \ldots t_{|w|} \in \Runs^{U,V,s,k}_{\B}(w)$ and $I_\rho=[l,r]$. Minimality of $r_s$ gives $r \ge r_s$, while $h=k$ gives $r \le r_s$. Thus $r=r_s$ and $t_{r_s} = (q_{k-1}, w_{r_s},p_k)$. Furthermore, $t_{r_s+1} \ldots t_{|w|}$ witnesses a run $p_k \trans{w[r_s+1..|w|]}q_k$. By \cref{claim:unambiguous} it is a unique run, so runs in $\Runs^{U,V,s,k}_{\B}(w)$ differ only in the prefix $p_1 \trans{w[1..r_s-1]} q_{k-1}$. Every such prefix belongs to $\Runs^{U,W}_{\B_{p_1 \to q_{k-1}}}(w[1..r_s-1])$ for some nonempty $W\subseteq Q_{k-1}$. There are at most $2^{n_{k-1}}$ choices of $W$, and for each choice there are at most $N[1,k-1]$ prefixes. Hence
\[
|\Runs^{U,V,s,k}_{\B}(w)| \le 2^{n_{k-1}}N[1,k-1].
\]

Combining everything we get
\begin{align*}
N[1,k] &= |\Runs^{U,V}_{\B}(w)| = \sum_{s=1}^{m} \sum_{h=2}^{k} |\Runs^{U,V,s,h}_{\B}(w)| \\
&\le m \left(\sum_{h=2}^{k-1} n_h2^{2n_h} \cdot N[1,h]\cdot N[h,k] + 2^{n_{k-1}}N[1,k-1] \right) \\
&\le 2m \cdot \sum_{h=2}^{k-1} n_h2^{2n_h} \cdot N[1,h]\cdot N[h,k],
\end{align*}
where the last inequality follows by bounding the additional term by the summand for $h=k-1$: we have $N[k-1,k]\ge1$ and $2^{n_{k-1}}\le n_{k-1}2^{2n_{k-1}}$. We conclude by \cref{lemma:fulldecomposition} as recall that $m\le 2^{\min(n_1,n_k)}$.
\end{proof}

\begin{proof}[Proof of \cref{theorem:main}]
Let $\A$ be a trim finitely ambiguous chain NFA, where $n_1$, \ldots $n_k$ are the sizes of its consecutive SCCs and $n = \sum_{i=1}^k n_i$. Let $Q_1$ and $Q_k$ be the first and last SCC, respectively. We assume that $k > 1$, otherwise we are done by \cref{claim:unambiguous}. Denote $\vec{n} = (n_1,\ldots,n_k)$.
Recall that $\Runs_{\A}(w) = \bigcup_{U \subseteq Q_1} \bigcup_{V \subseteq Q_k} \Runs^{U,V}_{\A}(w)$ and that $|\Runs^{U,V}_{\A}(w)| \le N_{\vec{n}}[1,k]$ for all words $w$ and choices of $U$ and $V$. 
Given that there are at most $2^{n_1+n_k}$ choices for $U$ and $V$ it suffices to show that $N_{\vec{n}}[1,k] \le 2^{\Oh(n)}$.

First we show by induction that 
\begin{align}\label{eq:induction}
N[i,j] \le 2^{2(j-i) - 1} \cdot 2^{\min(n_i,n_j)} \cdot C_{j-i-1} \cdot \prod_{h = i}^{j-1} N[h,h+1] \cdot \prod_{h = i+1}^{j-1} \bigl(n_h \cdot 2^{4n_h}\bigr),
\end{align}
where $C_{j-i-1}$ is the ($j-i-1$)-th Catalan number.
The induction is on $j-i$. Recall that $N_{\vec{n}}[i,j]$ is convenient notation for $N_{(n_i,\ldots,n_j)}[1,j-i+1]$ so we can apply \cref{lemma:recursion}.

The base case $j-i = 1$ follows directly from \cref{lemma:recursion} for $k=2$. For $j-i > 1$ we have by \cref{lemma:recursion}
\begin{align}\label{eq:nij}
N[i,j] \le 2 \cdot 2^{\min(n_i,n_j)} \cdot \sum_{h=i+1}^{j-1} n_{h} \cdot 2^{2 n_{h}} \cdot N[i,h]\cdot N[h,j].
\end{align}
By the induction assumption
\begin{align*}
N[i,h]\cdot N[h,j] \le \left( 2^{2(h-i) - 1} \cdot 2^{\min(n_i,n_h)} C_{h-i-1} \cdot \prod_{s = i}^{h-1} N[s,s+1] \cdot \prod_{s = i+1}^{h-1} n_{s} \cdot \prod_{s=i+1}^{h-1} 2^{4n_s}\right) \\
\cdot \left( 2^{2(j-h) - 1} \cdot 2^{\min(n_h,n_j)} C_{j-h-1} \cdot \prod_{s = h}^{j-1} N[s,s+1] \cdot \prod_{s = h+1}^{j-1} n_{s} \cdot \prod_{s=h+1}^{j-1} 2^{4n_s} \right) \\
= C_{h-i-1} C_{j-h-1} 2^{2(j-i) -2} \cdot 2^{\min(n_i,n_h) + \min(n_h,n_j)} \cdot \prod_{s = i}^{j-1} N[s,s+1] \cdot \prod_{s \in [i+1..j-1]\setminus \set{h}} n_{s} 2^{4n_s}.
\end{align*}
Observe that $2^{\min(n_i,n_h) + \min(n_h,n_j)} \le 2^{2n_h}$. Applying this in \cref{eq:nij} we get
\begin{align*}
N[i,j] \le 2^{2(j-i) - 1} \cdot 2^{\min(n_i,n_j)} \cdot \prod_{h = i}^{j-1} N[h,h+1] \cdot \prod_{h = i+1}^{j-1} \bigl(n_h \cdot 2^{4n_h}\bigr) \cdot \sum_{h=i+1}^{j-1} C_{h-i-1} C_{j-h-1}.
\end{align*}

To conclude let us shift indices in the sum by $i+1$ and denote $s = j-i-1$. Then the sum becomes $\sum_{h=0}^{s-1} C_{h} C_{s-1-h}$, which is equal to $C_{s} = C_{j-i-1}$ as required.

Consider \cref{eq:induction} for $i = 1$ and $j = k$.
We state three simple observations.
\begin{enumerate}
 \item\label{eq:1} $\prod_{h = 1}^{k-1} N[h,h+1] \le 2^{\sum_{h=1}^{k-1}\min(n_h,n_{h+1})} \le 2^{n_1 + \ldots + n_{k-1}} \le 2^n$.
 \item\label{eq:2} $\prod_{h = 1}^k n_{h} \le \prod_{h = 1}^k 2^{n_{h}} = 2^{n_1 + \ldots +n_k} = 2^{n}$.
 \item\label{eq:3} $C_k \le 4^k \le 4^n$.
\end{enumerate}
\cref{eq:1} follows from \cref{lemma:recursion} for $k=2$.
We conclude that $N[1,k] \le 2^{\Oh(n)}$.
\end{proof}

As a consequence, we obtain an upper bound for finite semigroups generated by zero-one matrices.
Matrix products below use ordinary integer addition and multiplication.

\begin{theorem}\label{theorem:matrices}
Let $n\ge1$ and let $S$ be a finite semigroup generated by $n\times n$ zero-one matrices.
Then every $A\in S$ satisfies
\[
 \sum_{p,q=1}^n A_{pq}\le 2^{\Oh(n)},
\]
where the implicit constant is independent of the generators.
\end{theorem}
\begin{proof}
Let $\A=\{A_1,\ldots,A_m\}$ be the set of generators.
Construct an NFA $\B$ with state set $Q=\{1,\ldots,n\}$, alphabet $\{a_1,\ldots,a_m\}$, and initial and final sets both equal to $Q$.
Include the transition $(p,a_i,q)$ if and only if $(A_i)_{pq}=1$.
For every nonempty word $w=a_{i_1}\cdots a_{i_\ell}$, matrix multiplication counts runs, so
\[
 |\Runs_{\B}(w)|=\sum_{p,q=1}^n(A_{i_1}\cdots A_{i_\ell})_{pq}.
\]
Since $S$ is finite and the empty word has exactly $n$ accepting runs, $\B$ is finitely ambiguous.
By \cref{lemma:chain,theorem:main}, it has at most $2^{\Oh(n)}$ accepting runs on every word.
Every element of $S$ is a nonempty product of generators, which gives the claimed bound.
\end{proof}

\section{Lower bounds}
\label{sec:lower}

Weber and Seidl~\cite[Theorem 5.1, Table 1]{WeberS91} give a lower bound of $2^{1.0221n}$ on the maximal finite ambiguity of $n$-state NFAs, for $n\equiv 0\pmod{64}$.
Their examples are chain NFAs with singleton SCCs over five letters; they attribute the construction to Lemma~5.7 of Weber's PhD thesis.
We improve this lower bound to $\Omega((1+\sqrt{2})^n)$.
Our automata still have five letters and singleton SCCs, but they are not chain NFAs as some transitions skip a state.

Let $P_0=0$, $P_1=1$, and $P_{m+2}=P_m+2P_{m+1}$ be the Pell numbers.
\begin{theorem}\label{theorem:lower}
For every $n\ge 3$, there is a trim finitely ambiguous NFA $\A_n$ with $n$ states, singleton SCCs, and an alphabet of five letters, and a word $u_n$ such that
\[
  |\Runs_{\A_n}(u_n)|\ge 2P_{n-2}
  =\frac{(1+\sqrt{2})^{n-2}-(1-\sqrt{2})^{n-2}}{\sqrt{2}}.
\]
In particular, the maximal ambiguity is $\Omega((1+\sqrt{2})^n)$, where $\log_2(1+\sqrt{2})=1.271553\ldots$.
\end{theorem}

Let us describe the construction.
First construct an auxiliary automaton $\B_n$ on $q_1,\ldots,q_n$, with initial state $q_1$ and final state $q_n$.
For each $1\le i\le n-2$, introduce two letters $a_i,b_i$ and add the transitions in \cref{fig:lower-gadget}.
The automaton is the union of these overlapping three-state gadgets.

\begin{nolinenumbers}
\begin{center}
\begin{minipage}{\linewidth}
\centering
\captionsetup{type=figure}
\begin{tikzpicture}[
  >={Stealth[length=1.5mm]},
  lower state/.style={circle,draw,minimum size=8mm,inner sep=1pt},
  lower count/.style={circle,draw,fill=white,minimum size=5mm,inner sep=0pt,font=\scriptsize},
  every node/.style={font=\small},
  every edge/.style={draw,->}]
  \node[lower state] (p) at (0,-1.48) {$q_i$};
  \node[lower state] (r) at (2,-1.48) {$q_{i+1}$};
  \node[lower state] (s) at (4,-1.48) {$q_{i+2}$};
  \path (p) edge[loop above] node {$a_i$} (p)
        (r) edge[loop above] node {$a_i,b_i$} (r)
        (s) edge[loop above] node {$b_i$} (s)
        (p) edge node[above] {$b_i$} (r)
        (r) edge node[above] {$a_i$} (s)
        (p) edge[bend right=27] node[below] {$a_i,b_i$} (s);

  \foreach \row in {1,...,5}{
    \draw[black!15,line width=.25pt] (5.72,{-0.74*(\row-1)}) -- (12.77,{-0.74*(\row-1)});
    \node[anchor=east,font=\scriptsize] at (5.62,{-0.74*(\row-1)}) {$q_{\row}$};
  }
  \foreach \col in {0,...,6}{
    \draw[black!10,line width=.25pt] ({6+1.07*\col},.27) -- ({6+1.07*\col},-3.23);
  }
  \foreach \col/\letter/\idx in {1/a/1,2/b/1,3/a/2,4/b/2,5/a/3,6/b/3}{
    \node[font=\scriptsize] at ({6+1.07*(\col-.5)},.55) {$\letter_{\idx}$};
  }
  \foreach \name/\col/\row/\count in {n01/0/1/1,n11/1/1/1,n13/1/3/1,n22/2/2/1,n23/2/3/2,n32/3/2/1,n33/3/3/2,n34/3/4/3,n43/4/3/3,n44/4/4/4,n53/5/3/3,n55/5/5/7}{
    \node[lower count] (\name) at ({6+1.07*\col},{-0.74*(\row-1)}) {$\count$};
  }
  \node[lower count,double,double distance=.7pt] (n65) at (12.42,-2.96) {$10$};
  \foreach \from/\to in {n01/n11,n01/n13,n11/n22,n11/n23,n13/n23,n22/n32,n22/n34,n23/n33,n23/n34,n32/n43,n32/n44,n33/n43,n34/n44,n43/n53,n43/n55,n44/n55,n53/n65,n55/n65}{
    \draw[->] (\from) -- (\to);
  }
  \node at (2,-3.65) {(a) Transition gadget};
  \node at (9.15,-3.65) {(b) Accepting runs on $w_5$};
\end{tikzpicture}
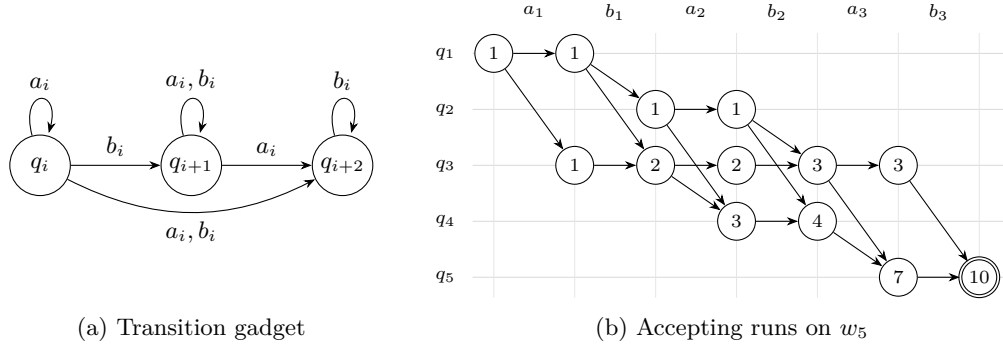
\caption{(a) The gadget for $a_i,b_i$ in $\B_n$, added for every $1\le i\le n-2$.
(b) The accepting runs of $\B_5$ on $w_5=a_1b_1a_2b_2a_3b_3$; each node is labelled with the number of partial runs reaching it.}
\label{fig:lower-gadget}
\end{minipage}
\end{center}
\end{nolinenumbers}

On $w_n=a_1b_1\cdots a_{n-2}b_{n-2}$, each subword $a_i b_i$ transforms the run counts $(x,y)$ at $q_i,q_{i+1}$ into $(x+y,2x+y)$ at $q_{i+1},q_{i+2}$.
Starting from $(1,0)$, after $a_1 b_1 \cdots a_i b_i$ we arrive at $(P_{i+1}-P_i, 2 P_i)$.
Figure~\ref{fig:lower-gadget}(b) illustrates $w_5$, which has $2P_3=10$ accepting runs.

Now take $\Sigma=\{c_0,\ldots,c_4\}$ and interpret all subscripts of $c$ modulo five.
Relabel $a_i$ by $c_i$ and $b_i$ by $c_{i+1}$ to obtain $\A_n$, and let $u_n=c_1c_2\,c_2c_3\cdots c_{n-2}c_{n-1}$.
All the accepting runs counted above survive.
Finite ambiguity follows from \cref{theorem:ida}: a word labelling a loop at $q_j$ uses only $c_{j-1},c_j$, and a run on these letters starting at $q_j$ cannot pass $q_{j+2}$.
Neither $q_{j+1}$ nor $q_{j+2}$ can complete an IDA witness.
The full proof is in \cref{app:lower}.


\bibliography{refs}
\clearpage
\appendix

\section{Proof of the lower bound}\label{app:lower}

\begin{proof}[Proof of \cref{theorem:lower}]
Fix $n\ge 3$.
We first define the auxiliary automaton $\B_n$ formally, then count its runs on $w_n$, and finally verify that relabelling its transitions gives a finitely ambiguous automaton over five letters.

\paragraph*{The auxiliary automaton.}
Let $Q=\{q_1,\ldots,q_n\}$ and $\Gamma_n=\{a_i,b_i\mid 1\le i\le n-2\}$.
Set $\B_n=(Q,\Gamma_n,\delta_n,q_1,q_n)$, where
\begin{align*}
\delta_n=\bigcup_{i=1}^{n-2}\{\,
 &(q_i,a_i,q_i),\ (q_{i+1},a_i,q_{i+1}),\
   (q_i,a_i,q_{i+2}),\ (q_{i+1},a_i,q_{i+2}),\\
 &(q_{i+1},b_i,q_{i+1}),\ (q_{i+2},b_i,q_{i+2}),\
   (q_i,b_i,q_{i+1}),\ (q_i,b_i,q_{i+2})\,\}.
\end{align*}
These are the transitions of \cref{fig:lower-gadget}.
Every transition either stays in its state or increases its index, so all SCCs are singletons.
There is a transition from $q_j$ to $q_{j+1}$ for every $1\le j<n$: use $b_j$ if $j\le n-2$, and $a_{n-2}$ if $j=n-1$.
Hence $\B_n$ is trim.

\paragraph*{Counting runs.}
For $0\le i\le n-2$, let $v_i=a_1b_1\cdots a_i b_i$, where $v_0=\epsilon$.
We show by induction on $i$ that all runs from $q_1$ on $v_i$ end in $\{q_{i+1},q_{i+2}\}$, and that their respective numbers $x_i,y_i$ satisfy
\begin{equation*}
  (x_i,y_i)=(P_{i+1}-P_i,2P_i).
\end{equation*}
For $i=0$, we have $(x_0,y_0)=(1,0)=(P_1-P_0,2P_0)$.
For the induction step, let $1\le i\le n-2$.
If the counts at $q_i,q_{i+1}$ before reading $a_i b_i$ are $(x,y)$, then reading $a_i$ gives counts $x,y,x+y$ at $q_i,q_{i+1},q_{i+2}$.
Reading $b_i$ then gives counts $x+y,2x+y$ at $q_{i+1},q_{i+2}$, with no other states reachable.
By the induction hypothesis, this yields
\[
  (x_i,y_i)=(x_{i-1}+y_{i-1},2x_{i-1}+y_{i-1})=(P_i+P_{i-1},2P_i)=(P_{i+1}-P_i,2P_i),
\]
where the last equality uses the Pell recurrence.
Since $v_{n-2}=w_n$ and $q_n$ is the only final state,
\begin{equation}\label{eq:lower-witness}
  |\Runs_{\B_n}(w_n)|=y_{n-2}=2P_{n-2}.
\end{equation}

\paragraph*{Reducing the alphabet.}
Let $\Sigma=\{c_0,\ldots,c_4\}$, using the convention $c_{j+5}=c_j$ for every integer $j$.
Define the letter map $h:\Gamma_n\to\Sigma$ by $h(a_i)=c_i$ and $h(b_i)=c_{i+1}$, and extend $h$ to words by concatenation.
Let $\A_n=(Q,\Sigma,\delta'_n,q_1,q_n)$, where
\[
 \delta'_n=\{(p,h(a),q)\mid(p,a,q)\in\delta_n\}.
\]
The underlying directed graph is unchanged, so $\A_n$ is trim and has singleton SCCs.
For a fixed word in $\Gamma_n^*$, a run is uniquely determined by its sequence of states.
Consequently, distinct runs of $\B_n$ on $w_n$ remain distinct runs of $\A_n$ on $u_n=h(w_n)$.
By \cref{eq:lower-witness}, $|\Runs_{\A_n}(u_n)|\ge 2P_{n-2}$.

\paragraph*{Finite ambiguity.}
Inspection of $\delta'_n$ gives the following possible labels; near the endpoints some of them may be absent:
\[
\begin{array}{c|c}
 \text{transition} & \text{possible labels} \\ \hline
 q_j\to q_j     & c_{j-1},c_j \\
 q_j\to q_{j+1} & c_{j-1},c_{j+1} \\
 q_j\to q_{j+2} & c_j,c_{j+1}
\end{array}
\]
Fix $j$ and put $L_j=\{c_{j-1},c_j\}$.
Every run from $q_j$ on a word in $L_j^*$ stays in $\{q_j,q_{j+1},q_{j+2}\}\cap Q$.
Indeed, a transition from $q_j$ increases the index by at most two.
From $q_{j+1}$, the only possible transition to a different state using $L_j$ goes to $q_{j+2}$, on $c_j$.
From $q_{j+2}$, no outgoing transition uses $L_j$: such labels belong to $\{c_{j+1},c_{j+2},c_{j+3}\}$, which is disjoint from $L_j$ because the alphabet has five letters.

Suppose that $\A_n$ complies with IDA.
There are distinct states $q_j,q_\ell$ and a word $v$ such that $q_j\trans{v}q_j$, $q_j\trans{v}q_\ell$, and $q_\ell\trans{v}q_\ell$.
State indices never decrease, so $\ell>j$, $v\ne\epsilon$, and the two runs returning to their starting states use only self-loops.
In particular, $v\in L_j^*$, and the preceding observation implies $\ell\in\{j+1,j+2\}$.
If $\ell=j+1$, all letters of $v$ lie in $L_j\cap L_{j+1}=\{c_j\}$.
But a run from $q_j$ on $c_j^*$ cannot reach $q_{j+1}$: no transition from $q_j$ to $q_{j+1}$ has label $c_j$, and skipping to $q_{j+2}$ precludes returning.
If $\ell=j+2$, then $L_j\cap L_{j+2}=\emptyset$, contradicting $v\ne\epsilon$.
Both cases are impossible.
By \cref{theorem:ida}, $\A_n$ is finitely ambiguous.

Finally, the characteristic roots of the Pell recurrence are $1+\sqrt{2}$ and $1-\sqrt{2}$.
Its initial values give
\[
 2P_{n-2}=\frac{(1+\sqrt{2})^{n-2}-(1-\sqrt{2})^{n-2}}{\sqrt{2}}
          =\Theta((1+\sqrt{2})^n),
\]
which proves the remaining assertions.
\end{proof}

\end{document}